\documentclass[lettersize,journal]{IEEEtran}
\usepackage{amsmath,amsfonts}
\usepackage{algorithmic}
\usepackage{algorithm}
\usepackage{array}
\usepackage{textcomp}
\usepackage{stfloats}
\usepackage{url}
\usepackage{verbatim}
\usepackage{graphicx}
\usepackage{cite}
\usepackage{subcaption}
\newtheorem{theorem}{Theorem}[section]

\newtheorem{proof}{Proof}

\usepackage{balance}

\begin{document}

\title{Leaf-centric Logical Topology Design for OCS-based GPU Clusters}

\author{
  Xinchi Han\IEEEauthorrefmark{6},
  \and Weihao Jiang\IEEEauthorrefmark{6}, 
  \and Yingming Mao\IEEEauthorrefmark{2},
  \and Yike Liu\IEEEauthorrefmark{2}, 
  \and Zhuoran Liu\IEEEauthorrefmark{6},
  \and Yongxi Lv\IEEEauthorrefmark{6},
  \and Peirui Cao\IEEEauthorrefmark{3},
  \and Zhuotao Liu\IEEEauthorrefmark{4},
  \and Ximeng Liu\IEEEauthorrefmark{6},
  \and Xinbing Wang\IEEEauthorrefmark{6},
  \and Changbo Wu\IEEEauthorrefmark{7}, 
  \and Zihan Zhu\IEEEauthorrefmark{8}, 
  \and Wu Dongchao\IEEEauthorrefmark{5},
  \and Yang Jian\IEEEauthorrefmark{5},
  \and Zhang Zhanbang\IEEEauthorrefmark{5},
  \and Yuansen Chen\IEEEauthorrefmark{5}, 
  \and Shizhen Zhao\textsuperscript{\IEEEauthorrefmark{6},*}\thanks{* Shizhen Zhao is the corresponding author.}\\ 
  \IEEEauthorblockA{
    \IEEEauthorrefmark{6} \textit{Shanghai Jiao Tong University, Shanghai, China}\\ 
    \IEEEauthorrefmark{2} \textit{Xi'an Jiao Tong University, Xi'an, China} \\
    \IEEEauthorrefmark{3} \textit{Nanjing University, Nanjing, China} \\
    \IEEEauthorrefmark{4} \textit{Tsinghua University, Beijing, China} \\
    \IEEEauthorrefmark{7} \textit{University of Science and Technology of China, Hefei, China \& Shanghai Innovation Institute, Shanghai, China} \\
    \IEEEauthorrefmark{8} \textit{University of Electronic Science and Technology of China, Chengdu, China} \\
    \IEEEauthorrefmark{5} \textit{Huawei, Dongguan, China}
  } \\
  \IEEEauthorblockA{
    hanxinchi@sjtu.edu.cn, weihao.jiang@sjtu.edu.cn, mao1234@stu.xjtu.edu.cn, \\
    cn-lyk@stu.xjtu.edu.cn, cocopromenade-9@sjtu.edu.cn, shjdblgklyx2435@sjtu.edu.cn, \\
    caopeirui@nju.edu.cn, zhuotaoliu@tsinghua.edu.cn, liuximeng@sjtu.edu.cn, \\
    xwang8@sjtu.edu.cn, wuchangbo@mail.ustc.edu.cn, 2021080911004@std.uestc.edu.cn, \\
    wudongchao@huawei.com, yangjian227@huawei.com, zhangzhanbang1@huawei.com, \\
    chengyuansen1@huawei.com, shizhenzhao@sjtu.edu.cn 
  }
}

\maketitle
\begin{abstract}

Recent years have witnessed the growing deployment of optical circuit switches (OCS) in commercial GPU clusters (e.g., Google’s A3 GPU cluster) optimized for machine learning (ML) workloads. Such clusters adopt a three-tier leaf–spine–OCS topology: servers attach to leaf-layer electronic packet switches (EPSes); these leaf switches aggregate into spine-layer EPSes to form a Pod; and multiple Pods are interconnected via core-layer OCSes. Unlike EPSes, OCSes only support circuit-based paths between directly connected spine switches, potentially inducing a phenomenon termed routing polarization, which refers to the scenario where the bandwidth requirements between specific pairs of Pods are unevenly fulfilled through links among different spine switches. The resulting imbalance induces traffic contention and bottlenecks on specific leaf-to-spine links, ultimately reducing ML training throughput.

To mitigate this issue, we introduce a leaf-centric paradigm to ensure traffic originating from the same leaf switch is evenly distributed across multiple spine switches with balanced loads. Through rigorous theoretical analysis, we establish a sufficient condition for avoiding routing polarization and propose a corresponding logical topology design algorithm with polynomial-time complexity. Large-scale simulations validate up to 19.27\% throughput improvement and a 99.16\% reduction in logical topology computation overhead compared to Mixed Integer Programming (MIP)-based methods.

\end{abstract}

\section{Introduction}
\begin{sloppypar}

\IEEEPARstart{I}{n} recent years, optical circuit switches (OCS) have been increasingly adopted in cluster architectures and vendors including Google, Huawei, and NVIDIA have deployed or are actively evaluating OCS for GPU clusters specifically optimized to handle machine learning (ML) workloads \cite{GoogleA3Supercomputers,google-kubernetes-ai-infra-integration,google-ai-hypercomputer-schedule-gke,Patronas:25,huawei_2024_all_optical_switch,huawei-optical-switch-intelligent-computing-2025}. Flows generated by ML workloads are distinguished by substantial data volumes and complies with the coflow property, where flows that complete transmission earlier may need to wait for those finishing later, thus rendering it highly sensitive to traffic contention \cite{han2025vclos,qian2024alibaba,dong2021accl}. However, unlike electronic packet switches (EPSes), an OCS can only establish circuit-based paths between directly connected devices. This results in fewer available routing paths for traffic traversing the OCSes. Consequently, inadequate OCS configuration may lead to the traffic contention caused by what we term the \textbf{routing polarization} issue.

The routing polarization problem is an inherently unique issue to OCS-based clusters that has \textbf{remained underexplored}. A typical OCS-based cluster \cite{poutievski2022jupiter,cao2021trod,GoogleA3Supercomputers,huawei_2024_all_optical_switch,huawei-optical-switch-intelligent-computing-2025,Patronas:25} adopts a three-tier leaf-spine-OCS architecture: specifically, servers are first connected to \emph{leaf-layer} EPSes, which are further interconnected with \emph{spine-layer} EPSes to form a single \textbf{Pod}, and multiple such Pods are interconnected via \emph{core-layer} OCSes. Notably, an OCS is transparent to network packets, thereby enabling dynamic adjustments to the number of inter-Pod interconnection links through OCS reconfiguration. Following prior work \cite{zhao2021understanding, poutievski2022jupiter}, we define the \textbf{Logical Topology} as the interconnection pattern of spine switches, formed by the configuration of the OCS core layer. Traditional Clos-based clusters \cite{singh2016jupiter,qian2024alibaba} employ EPSes in the core layer to interconnect spine switches. By contrast, OCSes in the core layer only support the establishment of one-to-one inter-spine connections, whereas EPSes enable many-to-many inter-spine connections. This constraint can induce the \textbf{Routing Polarization problem}\footnote{A well-recognized related issue is \textbf{Hash Polarization} \cite{qian2024alibaba}, defined as a phenomenon where traffic flows are unevenly distributed across multiple available paths due to hashing-based load-balancing decisions. Notably, Hash Polarization may coexist with Routing Polarization.}. Routing polarization refers to the scenario where the bandwidth requirements between specific pairs of Pods are unevenly fulfilled through links among different spine switches. Within a Pod, leaf switches are fully interconnected with spine switches via a limited number of links; as a result, the leaf-to-spine links connecting to spine switches with a higher density of established inter-Pod connections may become bottlenecks for cross-Pod communications. In contrast, in Clos-based networks, spine switches achieve full interconnection through core-layer EPSes, thereby avoiding such a issue. Existing works~\cite{2021Gemini,poutievski2022jupiter,10892202,cao2021trod,cao2023threshold} design logical topologies based on the inter-Pod bandwidth requirement, overlooking the link establishment details between spine switches. Consequently, they fail to mitigate the routing polarization problem, and we term such an approach the \textbf{Pod-centric} logical topology design paradigm.

To address this routing polarization problem, logical topology design must ensure that traffic originating from the same leaf switch is evenly distributed across multiple spine switches with balanced loads. We therefore refer to this alternative methodology as the \textbf{leaf-centric} logical topology design paradigm. This shift, however, is \textbf{non-trivial} for two critical reasons. First, the number of leaf switches far exceeds that of Pods. Adopting a leaf-centric perspective can lead to a significant increase in \textbf{computational overhead} when relying on Mixed Integer Programming (MIP) for logical topology computation, consistent with industry practices (e.g., Google \cite{poutievski2022jupiter}). Second, we find that a poorly designed intra-Pod \textbf{physical topology}, specifically, the wiring configuration between leaf and spine switches within the same Pod, not only renders the leaf-centric logical topology design an \textbf{NP-complete} problem but also induces \textbf{unavoidable} routing polarization. We further find that when the intra-Pod physical topology satisfies certain conditions, this problem is no longer NP-Complete, which allows us to propose a polynomial-time solution.

In this paper, we present \textbf{LumosCore}, a comprehensive framework that unifies a polynomial-time algorithm for leaf-centric logical topology design with the fundamental design principles of intra-Pod physical topologies. The principal contributions of this work are summarized as follows:

\begin{figure*}
    \centering
    \includegraphics[width=1\linewidth]{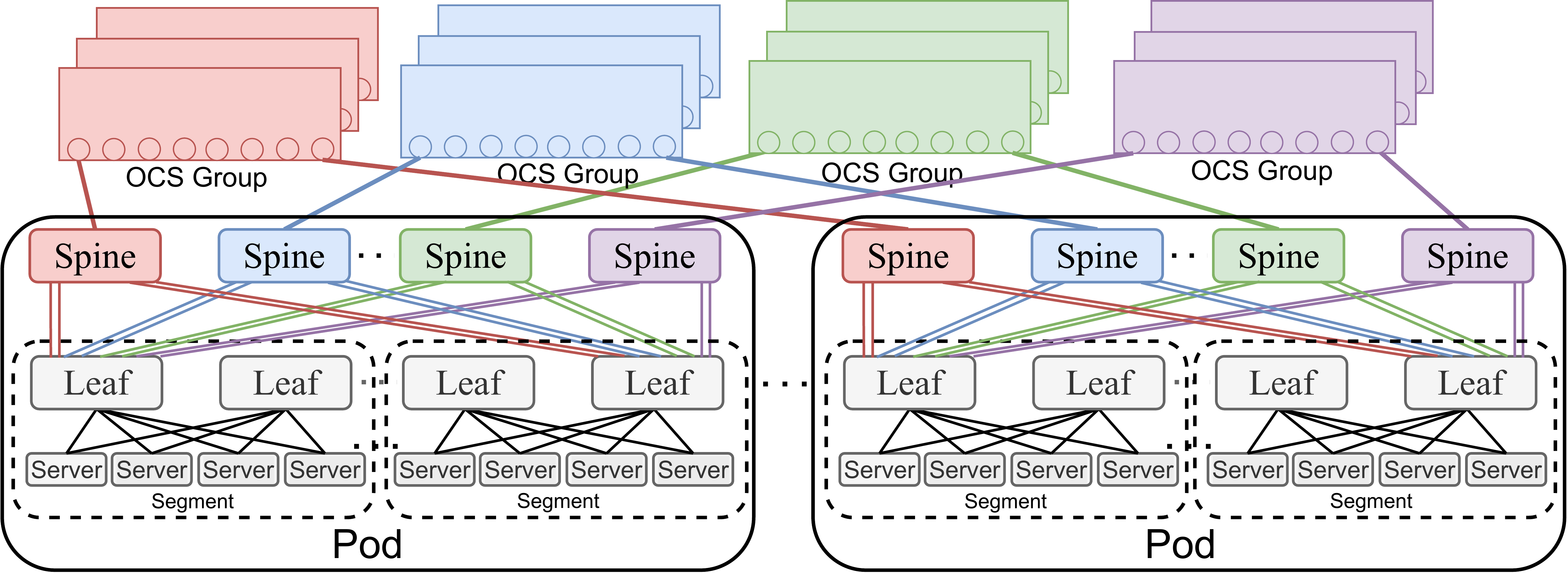}
    \caption{{LumosCore} adopts a typical three-tier topology comprising leaf layer, spine layer, and am optical core layer. }\label{fig:physical_architecture}
\end{figure*}

\begin{itemize}
\item In contrast to the previous Pod-centric logical topology design algorithm~\cite{2021Gemini,poutievski2022jupiter,cao2023threshold,10892202,0On}, we introduce a polynomial-time leaf-centric logical topology design algorithm specifically tailored to balance the traffic originating from the same leaf among multiple spine switches. 

\item Through rigorous theoretical analysis, we propose an intra-pod physical topology design and formally prove that, under this physical topology, the proposed polynomial-time algorithm can \textbf{avoid routing polarization}.

\item Using large-scale simulations driven by production workload traces, we further substantiate the superior performance of LumosCore, including up to a 19.27\% improvement in training throughput compared to Pod-centric strategy and a 99.16\% reduction in computational overhead for computing the logical topology compared to MIP-based strategy .

\end{itemize}
\end{sloppypar}

\section{Background and Motivation}
\begin{sloppypar}
\subsection{The basic architecture of LumosCore}
As illustrated in Fig.~\ref{fig:physical_architecture}, LumosCore adopts a three-tier hybrid leaf–spine–OCS network architecture, analogous to designs previously deployed in large-scale GPU clusters (e.g., Google’s A3/A4 OCS-based GPU clusters \cite{GoogleA3Supercomputers,google-ai-hypercomputer-schedule-gke,google-kubernetes-ai-infra-integration}).

\textbf{Intra-Pod Architecture:} Servers connect to the leaf switches using a commonly used Rail-Optimized configuration~\cite{qian2024alibaba}, enabling intra‑Segment (A \textbf{Segment} consists of servers interconnected via the same leaf switch) communication among servers via the leaf layer to reduce the inter-leaf traffic demand. Each leaf switch is provisioned with $K_{\text{leaf}}$ GPU-facing ports and an additional $K_{\text{leaf}}$ ports connected to the spine layer. More precisely, each leaf switch connects to $K_{\text{leaf}}/\tau$ distinct spine switches, where $\tau$ denotes the number of links between each leaf and each spine within a Pod. Consequently, a single pod contains $K_{\text{spine}}/\tau$ leaf switches and $K_{\text{leaf}}/\tau$ spine switches.

\textbf{Inter-Pod Architecture:} Since each pod incorporates $K_{\text{leaf}}/\tau$ spine switches, the architecture partitions all OCS devices into ${K_{\text{leaf}}}/{\tau}$ disjoint groups for achieving higher scalability. The $h$-th spine switch in each pod is connected to the $h$-th OCS group, thereby enforcing a consistent and deterministic mapping across all pods. Each OCS group comprises $K_{\text{spine}}$ OCS devices, and each OCS device is equipped with one pair of egress and ingress ports that connect to distinct pods. Let $K_{\text{ocs}}$ denote the total number of egress/ingress port pairs per OCS device. Under this construction, up to $K_{\text{ocs}}$ pods can be mutually interconnected.

\begin{figure*}
\centering
     
    \begin{subfigure}[t]{0.2\linewidth}
      \centering
      \includegraphics[width=\linewidth]{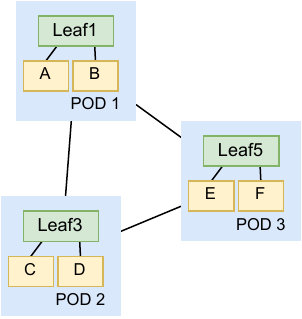}
      \caption{A Leaf-level Network Requirement example}
      \label{fig:example-logical}
     \end{subfigure}
     \begin{subfigure}[t]{0.39\linewidth}
     \centering
     \includegraphics[width=\linewidth]{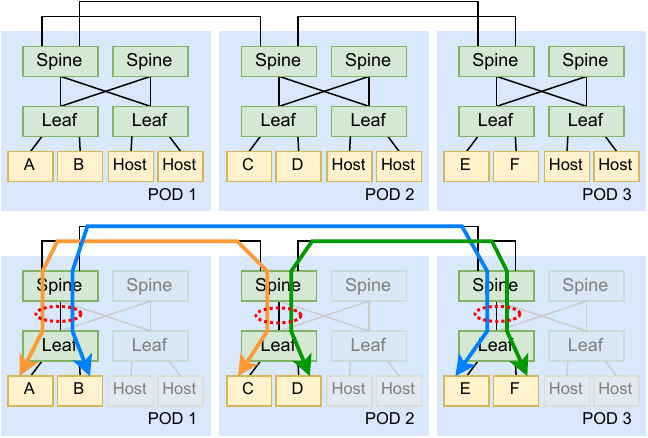}
     \caption{Pod-centric logical topology may introduce bottleneck on intra-Pod links}
     \label{fig:example-pod}
     \end{subfigure}
     \begin{subfigure}[t]{0.39\linewidth}
     \centering
     \includegraphics[width=\linewidth]{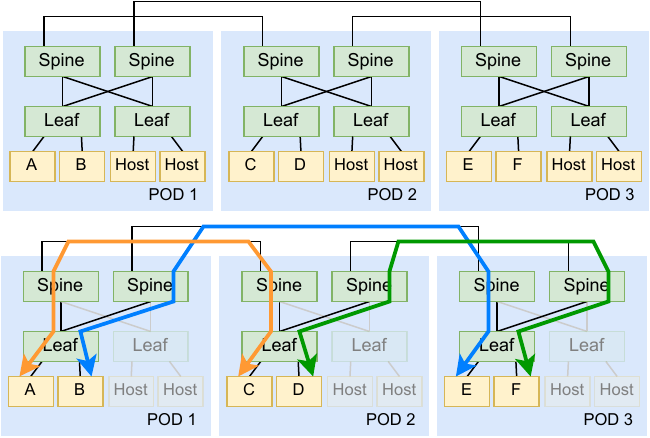}
     \caption{Leaf-centric logical topology alleviates such bottleneck by solving the routing polarization issue.}
     \label{fig:example-leaf}
    \end{subfigure}
    \emph{Upper figure shows the logical topology result under different paradigms, lower figure shows the routing result along with example flows under given logical topology. \\
    The bottleneck links, indicated by red circles, arise because the inter-Pod bandwidth requirements are unevenly fulfilled through links among different spine switches.}
    
    \caption{Pod-centric logical topology may result in the routing polarization issue.}
    \label{fig:example-conn}
\end{figure*}

\subsection{Important Concepts in LumosCore}
There are several important concepts in \emph{LumosCore} that will be used throughout the paper.\label{concept}

\begin{itemize}
    \item \textbf{\emph{Physical Topology}} characterizes the physical interconnection relationships among EPSes and between these switches and the OCSes. Existing work has predominantly examined how to interconnect EPSes with the OCSes \cite{zhao2021understanding,han2024lumoscore}. In contrast, this paper focuses on determining how \textbf{intra-Pod} leaf and spine switches should be interconnected within a three-tier leaf–spine–OCS architecture.


    \item \textbf{\emph{Leaf-level Network Requirement}} quantifies the number of disjoint cross-pod paths between every cross-pod leaf pair as Fig.\ref{fig:example-logical} shows. A \emph{Leaf-level Network Requirement} is determined by the bandwidth requirements of jobs need cross-pod communication. In certain cases, we allow multiple flows to share one inter-Pod path if the impact of such sharing on the corresponding jobs is minimal.

    \item \textbf{\emph{Logical Topology}} is the topology among different spine switches, formed by configuring the OCS layer. An OCS is transparent to network packets. Thus, creating a circuit inside an OCS for two spines is equivalent to directly adding a link between these two spines. \emph{Logical Topology} can be modified by reconfiguring the OCSes, which is \emph{a.k.a.} \emph{Topology Engineering~\cite{zhao2019minimal,poutievski2022jupiter,9651977}}. In OCS-based GPU clusers such as TopoOpt \cite{wang2022topoopt}, whenever a new task arrives, the logical topology must be reconfigured to satisfy the bandwidth requirements of the ML training workload.

    \item  \textbf{\emph{Routing Polarization}} refers to the scenario where the bandwidth requirements between specific pairs of Pods are unevenly fulfilled through links among different spine switches. Within a Pod, leaf switches are fully interconnected with spine switches through a limit number of links; consequently, the leaf-to-spine links connecting to spine switches with a higher density of established inter-Pod connections may become bottlenecks for cross-Pod communications. For instance, the logical topology in Fig.~\ref{fig:example-pod} satisfies the inter-Pod bandwidth requirements illustrated in Fig.~\ref{fig:example-logical}; however, traffic originating from Leaf 1 in Pod 1 can only be transmitted to both Pod 2 and Pod 3 through Spine 1. In this case, the intra-Pod link connecting to Spine 1 in Pod 1 become the communication bottleneck for cross-Pod traffic. It is noteworthy that \emph{Routing Polarization} may be \textbf{unavoidable} if the \emph{physical topology} is poorly designed. We present an example in Fig.~\ref{fig:logical topology design challenge}: given the \emph{Leaf-level Network Requirement} illustrated in Fig.~\ref{fig:logical topology}, a link is established between Leaf 1 of Pod 1 and Leaf 1 of Pod 2, and another between Leaf 1 of Pod 3 and Leaf 1 of Pod 2 (see Fig.~\ref{fig:illu_no_solution}). However, no additional free spine resources are available to establish a link between Leaf 1 of Pod 1 and Leaf 1 of Pod 3, resulting in the unavoidable routing polarization depicted in Fig.~\ref{fig:illu_protocal_incompatibility}.
    
\end{itemize}

\subsection{Requirements for Logical Topology Design Paradigms}\label{evaluation_of_logical}
Logical topology design paradigms should meet the following three requirements                         
First, Layer 2 (L2) protocol compatibility \textbf{must} be preserved. Specifically, if the ingress port of an OCS-facing port A on one spine switch is connected via the OCS to the egress port of an OCS-facing port B on another spine switch, then the egress port of A must correspondingly be connected via the OCS to the ingress port of B. This bidirectional connectivity constraint stems from the fact that standard L2 protocols (e.g., Address Resolution Protocol, ARP) unavoidably assume symmetric, bidirectional links. An equivalent constraint has also been imposed in prior OCS-based cluster deployments \cite{0On,10892202,poutievski2022jupiter,2021Gemini,han2024lumoscore}.

Second, the designed logical topology must avoid the routing polarization issue. Unlike conventional data center network (DCN) workloads, flows generated by ML workloads are characterized by large data volumes and comply with the coflow property, where flows completing transmission early may wait for those finishing later, making such traffic highly sensitive to congestion \cite{han2025vclos,qian2024alibaba,dong2021accl}.

Third, the logical topology design algorithm should exhibit polynomial-time computational complexity. While companies like Google leverage MIP to solve Pod-centric logical topology design problems \cite{poutievski2022jupiter}, the number of leaf switches far exceeds that of Pods, rendering the time overhead of MIP-based leaf-centric logical topology design prohibitive in large-scale clusters, particularly for ML workloads that require task-level recomputation of logical topologies \cite{2023TPU,wang2022topoopt}. Furthermore, MIP solvers fall into three primary categories: (1) proprietary solvers developed in-house by tech giants (e.g., Google, Huawei) \cite{perron2019or,HuaweiTCSLab_TAYLOR}; (2) open-source solvers (GLPK, LP-SOLVE, CBC) released under open-source licenses \cite{GNU2012GLPK,lpsolve_github,COINOR2005CBC}; and (3) commercial solvers requiring paid licenses (Gurobi, COPT, SCIP) \cite{Gurobi2024Downloads,CardinalOps2024COPT,SCIP2024Optimizer}. For small and medium-sized enterprises (SMEs) lacking the resources to build custom MIP solvers, adopting a polynomial-time algorithm for logical topology computation can significantly enhance the flexibility and practicality of deployment strategies.

\begin{figure*}
\centering
    \begin{subfigure}[b]{0.28\linewidth}
     \includegraphics[scale=0.6]{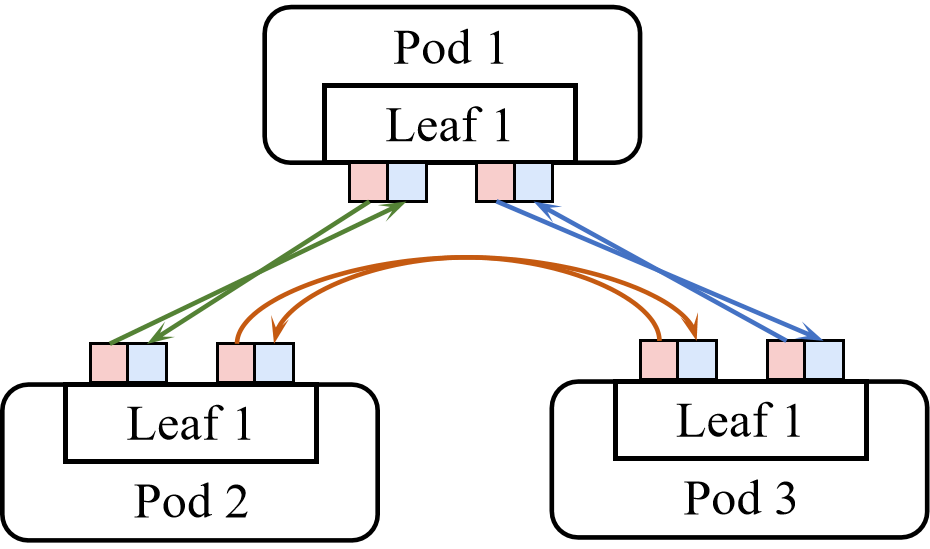}
     \caption{An example of Leaf-level Network Requirement, where leaf1 in three distinct Pods need to be interconnected with one bidirectional link.}
     \label{fig:logical topology}
     \end{subfigure}
     \begin{subfigure}[b]{0.32\linewidth}
     \includegraphics[scale=0.6]{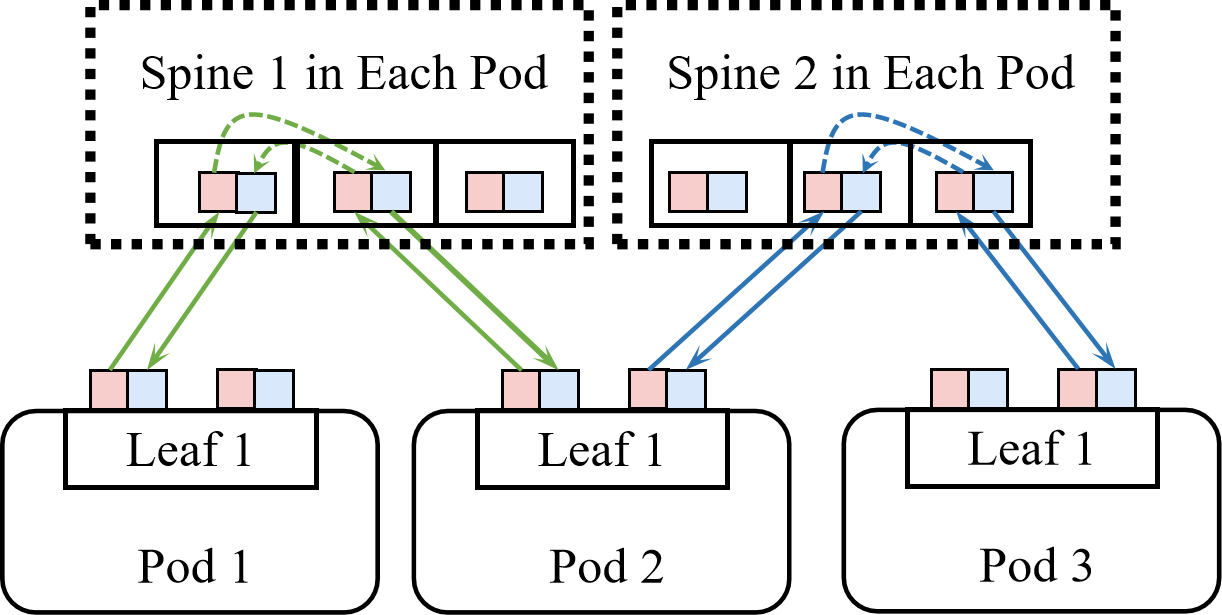}
     \caption{The design of the intra-Pod physical topology may result in unavoidable routing polarization. Note that Spine 1 and Spine 2 connect to different OCS Groups.}
     \label{fig:illu_no_solution}
     \end{subfigure}
     \begin{subfigure}[b]{0.29\linewidth}
     \includegraphics[scale=0.58]{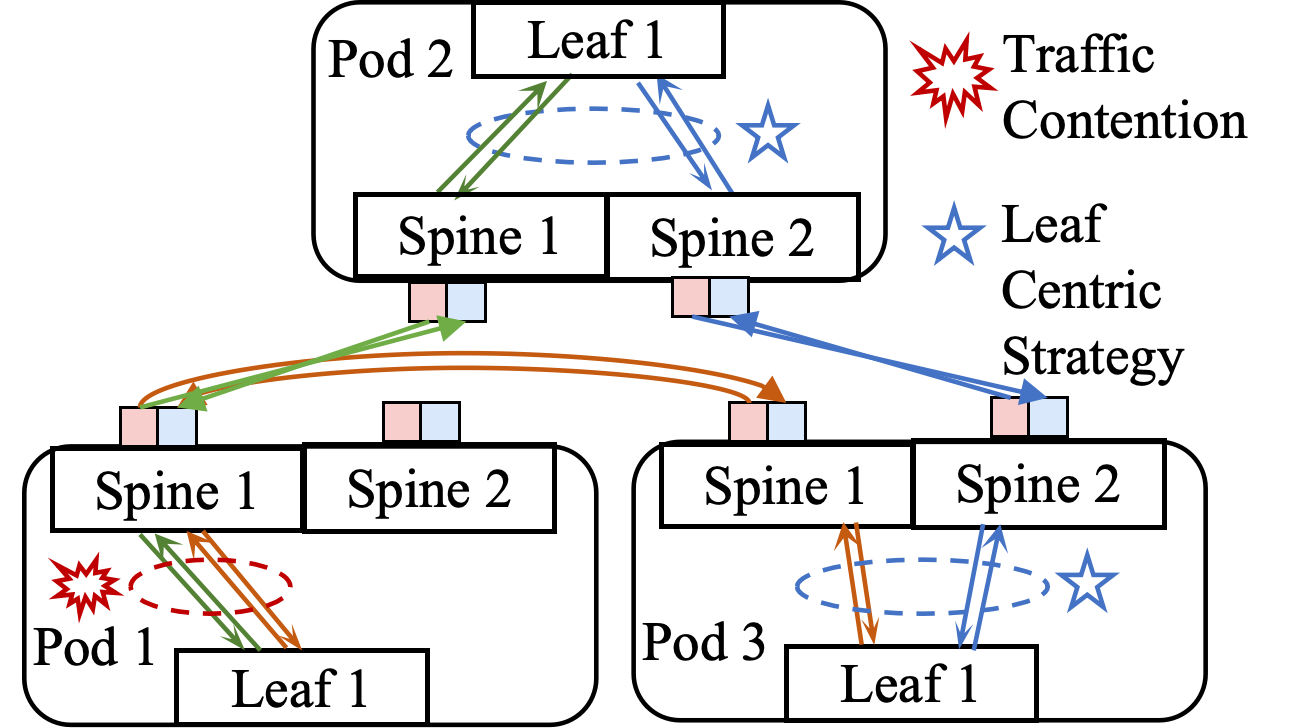}
     \caption{The given physical topology may result in unavoidable routing polarization, motivating us to further study how to design physical topology.}
     \label{fig:illu_protocal_incompatibility}
    \end{subfigure}
    \caption{An illustrative example demonstrating that a poorly designed intra-Pod physical topology can give rise to inherently unavoidable routing polarization. Notably, the logical topology should meet the L2 compatibility constraint \cite{poutievski2022jupiter,10892202,0On,2021Gemini,han2024lumoscore} which means that if the ingress port of an optical module $A$ on a spine switch is connected to the egress port of an optical module $B$ on another spine switch, then the egress port of $A$ must also be connected to the ingress port of $B$.}
    \label{fig:logical topology design challenge}
\end{figure*}

\subsection{The basic model for \textbf{Leaf-centric} Logical topology designing.}\label{modelPbuild}
We formulate an ILP (Integer Linear Programming) model to describe the designing of \textbf{Leaf-centric} logical topology as follows:

\noindent\textbf{Parameters}:
\begin{itemize}
    \item $P$: the number of Pods in a OCS-based GPU cluster.
    \item ${K_{spine}}$: the number of OCS-faced ports in each spine switch.
    \item ${K_{ocs}}$: the number of ports in each OCS.
    \item ${K_{leaf}}$: the number of spine-faced ports in each leaf switch.
    \item  $\tau$: the number of links between each leaf switch and each spine switch within a Pod.
    \item $L=[L_{ab}]$: The \emph{Leaf-level Network Requirement} $L$ is generated based on the communication demand of each GPU for all tasks, where $L_{ab}$ means the number of required links between the $a$-th leaf and the $b$-th leaf in the cluster.
\end{itemize}

\noindent\textbf{Decision Variables}:
\begin{itemize}
    \item $C_{ijh}$: The logical topology which means the number of connections between the $h$-th spine in the $i$-th Pod and the $j$-th Pod.
    \item $L_{abh}$: The number of connections between the $a$-th leaf and the $b$-th Leaf through the $h$-th spine in the $i$-th Pod and the $j$-th Pod, where $a\in \text{$i$-th Pod}$, $b\in \text{$j$-th Pod}$.
\end{itemize}
\textbf{Constraints}:

Given a \emph{Leaf-level Network Requirement Matrix} $L=[L_{ab}]$, where $L_{ab}=L_{ba}$ represents the inter-Pod network requirement between the $a$-th leaf and the $b$-th leaf and $L_{ab}=0$ if both leaves belong to the same Pod. According to the definition, $\sum_a L_{ab}\leq K_{leaf}$,$\sum_b L_{ab}\leq K_{leaf}$. We aim to find a \emph{Logical Topology} $C$ such that $L$ can be scheduled without contention\footnote{\textbf{By ``no contention'', we mean each leaf-level network requirement is fulfilled by a disjoint cross-Pod path. It is still possible for multiple flows to share one  intra-Pod path.}}.

Note that leaf-level network requirements can be fulfilled via different spines. Let $L_{abh}$ denote the number of required links from the $a$-th leaf to the $b$-th leaf that are fulfilled by the $h$-th spine. The sum of the numbers of required links from the $a$-th leaf to the $b$-th leaf across all spines must then equal the total number of required links $L_{ab}$ between these two leaves, i.e.,
\begin{equation}\label{sym_equ1}
   \sum_h L_{abh} = L_{ab}.
\end{equation}

Given $L_{abh}$, we can readily compute the total number of required links from the $a$-th leaf to the $h$-th spine as $\sum_b L_{abh}$, and that from the $h$-th spine to the $b$-th leaf as $\sum_a L_{abh}$. To \textbf{avoid routing polarization}, i.e., to prevent intra-Pod leaf-to-spine links from becoming communication bottlenecks, the aggregate network requirements $\sum_b L_{abh}$ and $\sum_a L_{abh}$ must satisfy:

\begin{equation}\label{sym_equ3}
\sum_b L_{abh} \leq \tau, \quad \sum_a L_{abh} \leq \tau.
\end{equation}

Given $L_{abh}$, we can compute $C_{ijh}$, the total number of required links in the logical topology between the $h$-th spine in the $i$-th Pod and the $h$-th spine in the $j$-th Pod as:

\begin{equation}\label{eqn:logical_topology}
C_{ijh}=\sum_{a\in \text{$i$-th Pod}}\sum_{b\in \text{$j$-th Pod}} L_{abh}.
\end{equation}
To ensure compatibility with L2 protocols, we enforce the following \emph{L2-compatibility constraint}:
\begin{equation}\label{sym_equ5}
   \sum_{a\in \text{$i$-th Pod}}\sum_{b\in \text{$j$-th Pod}} L_{abh} = \sum_{a\in \text{$i$-th Pod}}\sum_{b\in \text{$j$-th Pod}} L_{bah}. 
\end{equation}

The above constraints (\ref{sym_equ1}), (\ref{sym_equ3}), and (\ref{sym_equ5}) formulate an optimization model to derive the logical topology $C=[C_{ijh}]$ from the \emph{Leaf-level Network Requirement Matrix} $L$. This model can be solved using MIP solvers such as Gurobi \cite{achterberg2019s}; however, MIP incurs prohibitive computational costs, making it impractical for task-level OCS reconfiguration scenarios. We therefore need to develop a polynomial-time algorithm for logical topology design. Regrettably, we prove that the problem becomes NP-complete when the intra-Pod physical topology is ill-configured, as formally demonstrated in Theorem \ref{npc}.

\begin{theorem}\label{npc}
    The problem of designing a \textbf{Leaf-centric} logical topology is \emph{NP-complete} for intra-Pod physical topologies with $\tau=1$.
\end{theorem}
\begin{proof}
We prove the NP-completeness of the proposed model with $\tau=1$ via a reduction by restriction. Specifically, we show that the general case of our model encompasses a special case that is polynomial-time equivalent to the multi-coloring problem~\cite{halldorsson2004multicoloring,10892202}, a well-known NP-complete problem. By verifying this structural congruence and proving the NP-certificate property, we conclude the NP-completeness proof. We impose the following restrictions:
\begin{itemize}
    \item We set $\tau=1$,\emph{i.e.}, there exists one link between each leaf and each spine in each Pod.
    \item There exists certain cases where $\exists_{a,b} L_{ab}>0$.
\end{itemize}

The above constraints (\ref{sym_equ1})(\ref{sym_equ3})(\ref{sym_equ5}) can be transformed as follows:
\begin{align}
    \sum_{b} L_{abh} \leq 1, &\quad \forall a,h \label{new_eq2_5} \\
    \sum_{a} L_{abh} \leq 1, &\quad \forall b,h \label{new_eq3_5} \\
    L_{abh} = L_{bah},       &\quad \forall a,b,h \label{new_eq4} 
\end{align}

We consider a certain case where there exists at least a pair of $(a,b)$ so that the $a$-th leaf and $b$-th leaf need at least one link ($L_{ab} \geq 1,\exists a,b $). Whenever this condition holds, there must exist $a,b \in V$ satisfying the connectivity requirement:

\begin{equation}
    \sum_{h} L_{abh} \geq 1, \quad \forall a,b \in V \label{new_eq5}
\end{equation}

By constructing a graph transformation where $K_{\text{leaf}}$ represents the color palette size and each Pod pair connection $(a,b)$ is modeled as a virtual node, which implies Eqs. \eqref{new_eq4}, we establish correspondence with the multi-coloring problem. Virtual links connect two virtual nodes precisely when their corresponding Pod pair connection $(a,b)$ share common endpoints. Under this mapping, the constraint system \eqref{new_eq2_5}-\eqref{new_eq5} characterizes the generalized multi-coloring requirements: 
\begin{itemize}
    \item Each virtual node must receive at least one color (Eq. \eqref{new_eq5});
    \item Adjacent virtual nodes require distinct color assignments (Eqs. \eqref{new_eq2_5}\eqref{new_eq3_5}). 
\end{itemize}

The fundamental problem can be reduced by determining the existence of a valid coloring scheme using no more than $K_{\text{leaf}}$ colors, which is a classical NP-complete problem~\cite{halldorsson2004multicoloring}. By restriction method~\cite{garey1974some,10892202}, we prove the Theorem~\ref{npc}. This theoretical hardness, combined with the need for real-time computation, presents a significant challenge in practical deployment scenarios.

\end{proof}

\noindent\textbf{Discuss: }We formally prove that designing a logical topology is an NP-complete problem when $\tau = 1$. This observation motivates our analysis about the design of intra-Pod physical topology.

\subsection{Mathematical Preliminaries for \emph{LumosCore}}\label{Sec:symetric_matrix}
The design of \emph{LumosCore} is founded on two fundamental theorems: the Symmetric Matrix Decomposition Theorem and the Integer Matrix Decomposition Theorem. These theorems play a central role in topology design and in the development of polynomial-time complexity algorithms. The Symmetric Matrix Decomposition Theorem was originally formulated and proved in~\cite{han2024lumoscore}, while the Integer Matrix Decomposition Theorem was initially introduced and established in~\cite{zhao2018minimalextended}.

\begin{theorem}\label{lem:matrix_decomp}
\emph{(Symmetric Matrix Decomposition Theorem)} For any symmetric integer matrix $L$, there exists an integer matrix $A$, such that $L=A+A^T$ and $$\lfloor \frac{\sum_b L_{ab} }{2}\rfloor \leq \sum_b A_{ab} \leq \lceil \frac{\sum_b L_{ab} }{2}\rceil,\forall a.$$
$$\lfloor \frac{\sum_a L_{ab} }{2}\rfloor \leq \sum_a A_{ab} \leq \lceil \frac{\sum_a L_{ab} }{2}\rceil,\forall b.$$
\end{theorem}

\begin{theorem}\label{lem:minirewir}
\emph{(Integer Matrix Decomposition Theorem)} For any integer matrix $A$, there exists ${H}$ integer matrices, such that $A = A^{(1)}+ A^{(2)}+\dots+ A^{({H})}$, and for any $a=1,\dots,I$,$b=1,\dots,J$, $h=1,2,\dots,{H}$,
$$\lfloor\frac{A_{ab}}{{H}}\rfloor \leq A_{ab}^h \leq \lceil\frac{A_{ab}}{{H}}\rceil,$$
$$\lfloor\frac{\sum_aA_{ab}}{{H}}\rfloor \leq \sum_aA_{ab}^h \leq \lceil\frac{\sum_aA_{ab}}{{H}}\rceil,$$
$$\lfloor\frac{\sum_bA_{ab}}{{H}}\rfloor \leq \sum_bA_{ab}^h \leq \lceil\frac{\sum_bA_{ab}}{{H}}\rceil.$$
\end{theorem}

\end{sloppypar}

\section{Path to Leaf-centric Logical Topology Design Paradigm}\label{sec:logical_topo}
\begin{sloppypar}
Drawing upon the modeling framework introduced in $\S$\ref{modelPbuild}, it is evident that the configuration of the intra-Pod physical architecture directly impacts the synthesis of the logical topology design. This section therefore details the design rationale underlying the leaf-centric logical topology design paradigm and characterizes the specific intra-Pod physical architecture conditions under which the corresponding design problem admits a polynomial-time solution for any given valid \emph{Leaf-level Network Requirement Matrix}.

\subsection{Leaf-centric Logical Topology Design Paradigm}\label{spine_balance}

\subsubsection{A Polynomial-Time Algorithm for Logical Topology Design} \label{algorithm}
We propose a Heuristic-Decomposition algorithm that constructs a logical topology $C$ from the \emph{Leaf-level Network Requirement Matrix} $L$ with polynomial time complexity.

\noindent\textbf{Step 1:} By applying the Symmetric Matrix Decomposition Theorem (Theorem \ref{lem:matrix_decomp}), we decompose $L$ as $L = A + A^T$, where the matrix $A$ satisfies
\[
\sum_b A_{ab} \leq \Big\lceil \frac{\sum_b L_{ab}}{2} \Big\rceil, 
\qquad 
\sum_a A_{ab} \leq \Big\lceil \frac{\sum_a L_{ab}}{2} \Big\rceil.
\]

\noindent\textbf{Step 2:} Using the Integer Decomposition Theorem (Theorem \ref{lem:minirewir}), we further decompose $A$ into $H = K_{\text{leaf}}/\tau$ submatrices $A^{(1)}, A^{(2)}, \dots, A^{(H)}$ such that
\[
A = A^{(1)} + A^{(2)} + \cdots + A^{(H)},
\]
and, for each $h \in \{1,\dots,H\}$,
\[
\sum_{a} A_{ab}^h \leq \Big\lceil \frac{\sum_{a} A_{ab}}{H} \Big\rceil, 
\qquad 
\sum_{b} A_{ab}^h \leq \Big\lceil \frac{\sum_{b} A_{ab}}{H} \Big\rceil.
\]

\noindent\textbf{Step 3:} Define $L_{abh} = A_{ab}^h + A_{ba}^h$. The logical topology $C_{ijh}$ is then computed from $\{L_{abh}\}$ according to (\ref{eqn:logical_topology}).

\subsection{Time Complexity Analysis of Heuristic-Decomposition Algorithm}
Let the number of leaves in the cluster be denoted by $\eta = \frac{K_{\text{spine}} \times P}{\tau}$, where $P$ is the number of Pods. The time complexity of each step in Algorithm \ref{alg:heuristic_decomp} is analyzed as follows:
\begin{itemize}
  \item \textbf{Step 1}: Construct a multi-commodity flow (MCF) model following \cite{Z2012Efficient}, with time complexity $O(\eta^6 \log \eta)$.
  \item \textbf{Step 2}: Construct an MCF model as described in \cite{zhao2019minimal}, with time complexity $O(K_{\text{leaf }} \eta^4 \log \eta)$.
  \item \textbf{Step 3}: Direct computation of $L_{abh}$ and logical topology $C_{ijh}$, with time complexity $O(K_{\text{leaf }} \eta^2)$.
\end{itemize}
Since $K_{\text{leaf}} < \eta$, by applying the Master Theorem \cite{bentley1980general}, the overall time complexity of the Heuristic-Decomposition algorithm is dominated by Step 1 and is therefore $O(\eta^{6}\log \eta)$.

\begin{algorithm}[htbp]
  \caption{Heuristic-Decomposition for Leaf-centric Logical Topology Design}
  \label{alg:heuristic_decomp}
  \begin{algorithmic}[1] 
    \REQUIRE Leaf-level Network Requirement Matrix $L$, $K_{\text{spine}}$, $P$, $\tau$, $K_{\text{leaf}}$
    \ENSURE Logical topology $C = \{C_{ijh}\}$
    
    \STATE \textbf{Step 1: Symmetric Matrix Decomposition of $L$}
    \STATE Compute $A$ by decomposing $L$ so that $ A + A^T = L$ (via Symmetric Matrix Decomposition Theorem \ref{lem:matrix_decomp}) 
    
    \STATE \textbf{Step 2: Integer Decomposition of Matrix $A$}
    \STATE Compute $H = \frac{K_{\text{leaf}}}{\tau}$
    \STATE Compute $A^{(h)}$ by decomposing $A$ so that $\sum_{h=1}^H A^{(h)}=A$ (via Integer Decomposition Theorem \ref{lem:minirewir})
    
    \STATE \textbf{Step 3: Construct Logical Topology}
    \STATE For all leaf nodes $a,b$ and $h \in \{1, \dots, H\}$, compute $L_{abh} = A_{ab}^{(h)} + A_{ba}^{(h)}$
    \STATE For all Pod indices $i,j$ and $h \in \{1, \dots, H\}$, compute logical topology via:
    \STATE \quad $C_{ijh}=\sum_{a\in \text{$i$-th Pod}}\sum_{b\in \text{$j$-th Pod}} L_{abh}$ (via Eq.\eqref{eqn:logical_topology})
    
    \RETURN $C$
  \end{algorithmic}
\end{algorithm}

\subsection{How to design intra-Pod architecture}
When designing the intra-Pod architecture, we follow the leaf-spine structure commonly adopted in commercial cluster deployments \cite{qian2024alibaba,poutievski2022jupiter,DGX2023}, where leaf and spine switches are interconnected through a uniform full-mesh topology, as shown in Fig.\ref{fig:physical_architecture}. In this section, we concentrate on the selection of the parameter $\tau$, which denotes the number of links established between leaf and spine switches within a single Pod.

When $\tau = 2$, implying that each leaf switch is connected to each spine switch within the same Pod via two parallel links, it follows directly from the inequalities established in \textbf{Step 1} and \textbf{Step 2} in $\S$ \ref{algorithm} that the solution produced by the Heuristic-Decomposition algorithm satisfies constraints (\ref{sym_equ1}), (\ref{sym_equ3}), and (\ref{sym_equ5}). This result is formally stated in Theorem \ref{lem:contention_free_tau2}. Given an OCS with a port count of $K_{ocs}$, a three-tier OCS-based cluster can accommodate at most $K_{leaf} * K_{spine} * K_{ocs} / {\tau}$ GPUs. Thus, a smaller value of $\tau$ allows the cluster to host a larger number of GPUs. However, setting $\tau = 1$ may lead to the routing polarization issue. Therefore, according to Theorem \ref{lem:contention_free_tau2}, we set $\tau = 2$ when designing intra-Pod architecture in LumosCore.

\begin{theorem}\label{lem:contention_free_tau2}
When $\tau = 2$, for any Leaf-level Network Requirement Matrix $L$, the solution produced by the Heuristic-Decomposition algorithm satisfies all constraints (\ref{sym_equ1}), (\ref{sym_equ3}), and (\ref{sym_equ5}). Equivalently, there exists a logical topology $L$ such that $L$ can be scheduled without incurring any the routing polarization issue.
\end{theorem}

\begin{proof}
Consider a leaf-level network demand matrix \(L\), where \(L_{a,b}\) denotes the traffic requirement from the \(a\)-th leaf node to the \(b\)-th leaf node. By the theory of symmetric matrix decomposition, there exists a matrix \(A\) such that
\[
L_{a,b} = A_{a b} + A_{b a},
\]
for all leaf indices \(a, b\).

Moreover, the row and column sums of \(A\) are constrained by the following bounds:
\[
\left\lfloor \frac{\sum_{a} L_{a,b}}{2} \right\rfloor \leq \sum_{a} A_{a b} \leq \left\lceil \frac{\sum_{a} L_{a,b}}{2} \right\rceil,
\]
for every column index \(b\), and
\[
\left\lfloor \frac{\sum_{b} L_{a,b}}{2} \right\rfloor \leq \sum_{b} A_{a b} \leq \left\lceil \frac{\sum_{b} L_{a,b}}{2} \right\rceil,
\]
for every row index \(a\).  

These inequalities ensure that the aggregated incoming and outgoing components encoded in \(A\) approximate, up to a rounding of at most one unit, half of the corresponding total demand specified by \(L\).

By applying Integer Matrix Decomposition Theorem on $ A $, we can further decompose it into a three-dimensional tensor $ A_{ab}^h $ satisfying:

\begin{equation*}
    \sum_h A_{ab}^h = A_{ab}
\end{equation*}

\begin{equation*}
    \left\lfloor \frac{A_{ab}}{K_{\text{leaf}} / \tau} \right\rfloor \leq A_{ab}^h \leq \left\lceil \frac{A_{ab}}{K_{\text{leaf}} / \tau} \right\rceil
\end{equation*}

\begin{equation*}
    \left\lfloor \frac{\sum_a A_{ab}}{K_{\text{leaf}} / \tau} \right\rfloor \leq \sum_a A_{ab}^h \leq \left\lceil \frac{\sum_a A_{ab}}{K_{\text{leaf}} / \tau} \right\rceil
\end{equation*}

\begin{equation*}
    \left\lfloor \frac{\sum_b A_{ab}}{K_{\text{leaf}} / \tau} \right\rfloor \leq \sum_b A_{ab}^h \leq \left\lceil \frac{\sum_b A_{ab}}{K_{\text{leaf}} / \tau} \right\rceil
\end{equation*}

We then define $ L_{abh} $ as:

\begin{equation*}
    L_{abh} = A_{ab}^h + A_{ba}^h
\end{equation*}

To complete the proof of Theorem~\ref{lem:contention_free_tau2}, we verify that $ L_{abh} $ satisfies constraints (\ref{sym_equ1}), (\ref{sym_equ3}), and (\ref{sym_equ5}).

\textbf{Constraint (\ref{sym_equ1}) Verification:}

\begin{equation*}
    \sum_h L_{abh} = \sum_h (A_{ab}^h + A_{ba}^h) = A_{ab} + A_{ba} = L_{a,b}
\end{equation*}

Thus, constraint (\ref{sym_equ1}) is satisfied.

\textbf{Constraint (\ref{sym_equ3}) Verification:}

\begin{equation*}
    \sum_a L_{abh} = \sum_a (A_{ab}^h + A_{ba}^h)
\end{equation*}

Using the bounds derived earlier:

\begin{equation*}
    \sum_a L_{abh} \leq \left\lceil \frac{\sum_a A_{ab}}{K_{\text{leaf}} / \tau} \right\rceil + \left\lceil \frac{\sum_a A_{ba}}{K_{\text{leaf}} / \tau} \right\rceil
\end{equation*}

Since $ \sum_a A_{ab} + \sum_a A_{ba} \leq \sum_a L_{a,b} $, we can derive:

\begin{equation*}
    \sum_a L_{abh} \leq 2 \cdot \left\lceil \frac{\sum_a L_{a,b}}{2 \cdot K_{\text{leaf}} / \tau} \right\rceil
\end{equation*}

Similarly, we have:

\begin{equation*}
    \sum_b L_{abh} \leq 2 \cdot \left\lceil \frac{\sum_b L_{a,b}}{2 \cdot K_{\text{leaf}} / \tau} \right\rceil
\end{equation*}

Clearly, when $\tau = 2$, we can derive:

\begin{equation*}
    \sum_a L_{abh} \leq 2
\end{equation*}

\begin{equation*}
    \sum_b L_{abh} \leq 2
\end{equation*}

This confirms that constraint (\ref{sym_equ3}) holds when $\tau=2$.

\textbf{Constraint (\ref{sym_equ5}) Verification:}

\begin{equation*}
    \sum_a \sum_b L_{abh} = \sum_a \sum_b (A_{ab}^h + A_{ba}^h) = \sum_a \sum_b L_{bah}
\end{equation*}

Thus, constraint (\ref{sym_equ5}) is satisfied.

Since $ L_{abh} $ satisfies constraints (\ref{sym_equ1}), (\ref{sym_equ3}), and (\ref{sym_equ5}), Theorem \ref{lem:contention_free_tau2} is proved.
\end{proof}

\noindent\textbf{Remark: Handle the case when $\tau$=1} 

In certain scenarios, for instance, when a Pod has already been deployed and alterations to its topology are impractical, or when the cluster must scale to a very large size, it is still necessary to consider the setting in which the network is configured with $\tau = 1$. Under this configuration, the Heuristic-Decomposition algorithm described above can guarantee a maximum contention level of at most $2$, as Inequality~(\ref{sym_equ3}) is no longer satisfied. This upper bound of $2$ is in fact tight, as we have identified multiple instances of $L$ for which the model (\ref{sym_equ1})(\ref{sym_equ3})(\ref{sym_equ5}) admits no feasible solution when $\tau = 1$ (see Fig.~\ref{fig:logical topology design challenge}). Nonetheless, as established by Theorem~\ref{lem:contention_free_tau1}, contention-free schedules remain attainable provided that additional structural constraints are imposed on $L$.

\begin{theorem}\label{lem:contention_free_tau1}
For the special case $\tau = 1$, consider any matrix $L$ such that $\forall a,\; \sum_b L_{ab} \leq \frac{K_{\mathrm{leaf}}/\tau}{2}$ and $\forall b,\; \sum_a L_{ab} \leq \frac{K_{\mathrm{leaf}}/\tau}{2}$. Under these conditions, the model of equations \textup{(\ref{sym_equ1})}, \textup{(\ref{sym_equ3})}, and \textup{(\ref{sym_equ5})} admits a feasible solution that can be computed in time polynomial time.
\end{theorem}

Theorem~\ref{lem:contention_free_tau1} demonstrates that, under $\tau = 1$, minimizing cross-pod, inter-leaf communication demand is a critical design objective for specifying Leaf-level Network Requirements. This objective can be realized through judicious GPU resource scheduling. 

To establish Theorem~\ref{lem:contention_free_tau1}, we consider a simple greedy assignment procedure. When the conditions stated in Theorem~\ref{lem:contention_free_tau1} are satisfied, we can greedily assign each leaf-level network requirement to an unoccupied spine switch. Concretely, suppose a network requirement must be satisfied between leaf switch $a$ and leaf switch $b$. If leaf switch $a$ has previously utilized at most $\frac{K_{\text{leaf}}}{2} - 1$ distinct spine switches for communication, and the same holds for leaf switch $b$, then there exist at least two spine switches that are simultaneously available to satisfy the network requirement between leaf switches $a$ and $b$. Consequently, such a greedy strategy ensures a feasible contention-free assignment. The time complexity of this greedy algorithm is $O(K_{\text{leaf}} \times \eta)$.

\end{sloppypar}

\section{Large Scale Simulation}
\begin{sloppypar}
\subsection{Simulation Setup}
\label{sim_setup}
To address the computational inefficiency of packet-level fine-grained network simulators \cite{Kassing2016Netbench}, which can take days to simulate a single ML task trace on a cluster contains 64 GPUs, we adopt RapidAISim, a coarse-grained flow-level simulator specifically tailored for OCS-based GPU clusters.


We conduct a comprehensive performance evaluation across four cluster scales (2,048-GPU, 4,096-GPU, 8,192-GPU, and 16,384-GPU) to assess scalability and efficiency. The OCS-based clusters integrate 32-port EPSes with 256-port MEMS-OCS. Our comparative evaluation includes: (1) \textbf{Leaf-centric} approach with $\tau=2$ and $\tau=1$; (2) the \textbf{Pod-centric} approach used in works like Jupiter Evolving \cite{poutievski2022jupiter,2021Gemini}, which leverages inter-Pod link demands and MIP (with the high-performance Gurobi optimization library \cite{achterberg2019s} adopted as the MIP solver) to generate logical topologies; (3) the widely adopted 3-tier Clos architecture without oversubscription \cite{hu2021characterization,mahajan2020themis,gao2024crescent,qian2024alibaba} (equipped with common EPSes using Broadcom BCM56980 \cite{BroadcomBCM56980DS} switch chip which provides 12.8 Tbps switch capacity); and (4) A classic OCS-based cluster Helios \cite{farrington2010helios} which utilizes bipartite graph matching of traffic features for ToE. We make the following hardware assumptions: eight GPUs per server are interconnected via an intra-node fabric with 400 Gbps aggregate bandwidth, and inter-node GPU communication defaults to a 200 Gbps RDMA network \cite{eshelman2020dgx,GoogleA3Supercomputers}.

We generate 1000 ML tasks based on the SenseTime dataset \cite{hu2021characterization}, with scheduling constraints: Tensor Parallelism (TP) traffic is confined to a single server, and Expert Parallelism (EP) traffic is restricted within a single Pod. Task arrival intervals follow a Poisson distribution, with adjustments to ensure comparable workloads across architectures. We quantify cluster load using the \emph{workload level}, calculated via Equation (\ref{sim:1}), where $\lambda_k$ denotes the arrival rate of jobs requiring $k$ GPUs and $T_k$ represents their average runtime. Notably, $k \times \lambda_k \times T_k$ corresponds to the expected GPU time occupied by jobs requiring $k$ GPUs. We set the workload level to 0.767 by default. Equal-cost multi-path routing (ECMP) \cite{dhaliwal2018load} is adopted as the default load-balancing strategy. We adopt the standard MurmurHash3 \cite{senumammh3} algorithm as the hash function, and use the five-tuple (source IP address, destination IP address, source port, destination port, and transport layer protocol) as the hash factor.

  \begin{equation}\label{sim:1}
    \mbox{Workload-level} = \textstyle \sum_k \frac{k*\lambda_k*T_k}{GPU_{Num}} 
\end{equation}

\begin{figure}[t]
    \centering
    \includegraphics[width=1\linewidth]{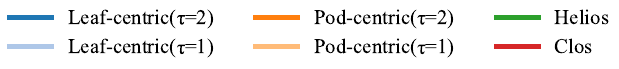}
    \begin{subfigure}[b]{0.49\linewidth}
     \includegraphics[width=\linewidth]{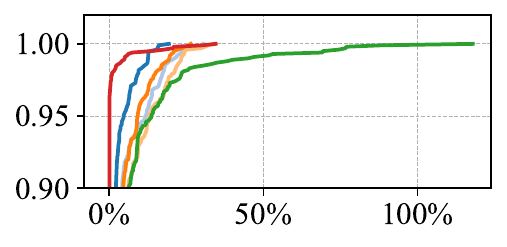}
     \caption{CDF of JRT slow down ratio compared to \emph{Best}}\label{case1}
    \end{subfigure}
    \begin{subfigure}[b]{0.49\linewidth}
        \centering  
    
        \includegraphics[width=\textwidth]{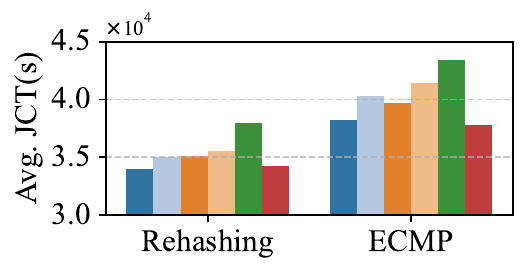}
    
        \caption{Performance under different load balance strategies}\label{fig:loadbalance}
    \end{subfigure}
    \begin{subfigure}[b]{0.49\linewidth}
     \includegraphics[width=\linewidth]{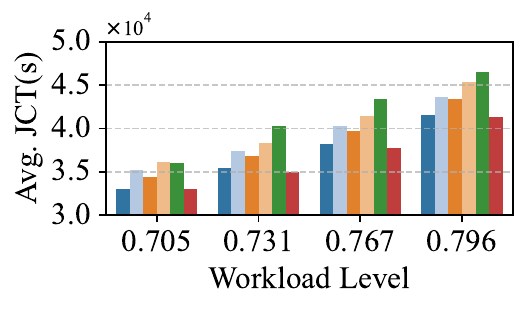}
     \caption{Avg.JCT under different workload-level.}\label{diff_jct_workload}
    \end{subfigure}
    \begin{subfigure}[b]{0.49\linewidth}
     \includegraphics[width=0.99\linewidth]
     {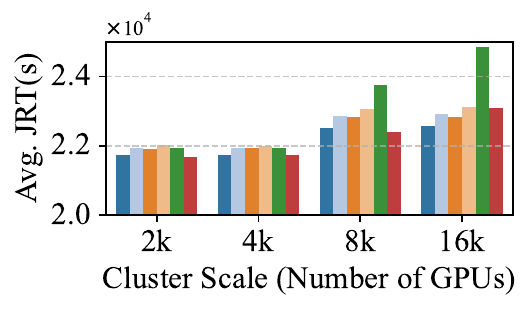}
    \caption{Avg.JRT under different cluster size.}\label{Avg.Bandwidth_vs_Node_Number}
    \end{subfigure}

\caption{The comparative performance of the evaluated strategies highlights the critical role of leaf-centric logical topology design and intra-Pod physical topology design. Unless otherwise specified, the workload level is fixed at 0.767, and the cluster configuration comprises 8,192 GPUs by default.}
\end{figure}

Performance evaluation is conducted through three key metrics: average job runtime ($Avg. JRT$), and average job completion time ($Avg. JCT$), where for each job $JRT = T_f-T_s$, and $JCT = T_f-T_a$ with $T_a$, $T_s$, and $T_f$ representing the job's arrival time, initiation time, and termination time respectively. To establish a theoretical upper bound for cluster performance assessment, we propose a hypothetical network topology featuring an idealized spine switch with unlimited port capacity, directly interconnecting all leaves under the \textbf{Best} architecture. The JRT values achieved in this optimal scenario are denoted as $JRT^{Best}$ for comparative analysis, while the performance degradation ratios are calculated as $\frac{JRT-JRT^{Best}}{JRT^{Best}}$ for the slow down ratio of JRT of each task.

\subsection{Performance Analysis}

Fig.~\ref{case1} presents the CDF of the JRT slowdown ratio across different strategies for an 8k-scale GPU cluster. Notably, the \textbf{Leaf-centric strategy with $\tau = 2$} outperforms all other approaches: compared to the Pod-centric strategy ($\tau = 2$), it achieves a maximum JRT reduction of 19.27\%, with 4\% of jobs experiencing a substantial JRT improvement (exceeding 5\%). Given that most jobs in the public dataset \cite{hu2021characterization} are small and do \emph{not involve} cross-Pod communication, we further focus on large jobs requiring \textbf{cross-Pod communication}: for this subset, the average JRT is reduced by 2.34\%, and 16.67\% of jobs achieve a significant JRT reduction (exceeding 5\%). This underscores that the Leaf-centric strategy effectively avoid routing polarization.

We further compare the Leaf-centric strategy under different intra-Pod configurations ($\tau = 1$ vs. $\tau = 2$). For jobs with cross-Pod communication, Leaf-centric ($\tau = 2$) achieves a maximum JRT reduction of 13.98\% relative to Leaf-centric ($\tau = 1$), with 19.44\% of tasks experiencing a JRT reduction exceeding 5\%. This demonstrates that intra-Pod architectural design exerts a significant impact on overall cluster performance, when $\tau = 1$, constructing an effective leaf-centric logical topology may be infeasible. It is worth noting that for suboptimal physical topology designs with $\tau=1$, while the Leaf-centric strategy still provides marginal benefits, traffic contention remains a critical issue, this further emphasizes the paramount importance of rational physical topology design.

\emph{Helios} \cite{farrington2010helios} exhibits suboptimal performance in large-scale clusters or at high workload levels, a phenomenon attributable to the increased complexity of logical topology computation under intricate traffic characteristics, simple greedy-based strategies thus struggle to cope with such scenarios effectively. Furthermore,  Leaf-centric ($\tau = 2$) delivers performance comparable to the \textbf{Clos} fabric while incurring substantially lower costs. The long-tail effect of JRT in Leaf-centric ($\tau = 2$) is significantly shorter, primarily because the Leaf-centric ($\tau = 2$) architecture only employs two tiers of EPSes, whereas the Clos fabric incorporates more tiers, thus incurring fewer hash function invocations for the former. Conversely, a greater number of hash tiers can induce more severe \textbf{Hash Polarization} \cite{qian2024alibaba}.

As illustrated in Fig.~\ref{diff_jct_workload}, the proposed strategy yields even more pronounced improvements in average Avg. JCT compared to Avg. JRT. This aligns with queueing theory \cite{adan2002queueing}, which posits that reductions in task queuing time can be substantially larger than those in task running time.

\subsection{Robustness Analysis}

\textbf{Performance under different load balancing strategies:} Traffic load balancing strategies exert a profound impact on traffic contention caused by hash polarization. Beyond the widely adopted ECMP strategy, we additionally incorporate a rehash-based ECMP scheme (implemented in Alibaba’s ACCL library \cite{dong2021accl}, a strategy that mitigates the \textbf{Hash Polarization} problem by performing multiple rounds of hashing and selecting the least congested path, denoted as \emph{Rehashing}). As shown in Fig.~\ref{fig:loadbalance}, more effective load balancing strategies yield reductions in average JRT across all approaches; however, the Leaf-centric method with $\tau = 2$ remains superior to the other other OCS-based cluster designs.

\textbf{Performance under different workload-level:} Workload exerts a notable impact on the performance of different strategies. To quantify this effect, we generate distinct datasets by adjusting $\lambda_k$ (the arrival rate of jobs requiring $k$ GPUs). As shown in Fig.~\ref{diff_jct_workload}, the {Leaf-centric strategy with $\tau = 2$} consistently outperforms all other OCS-based cluster designs across datasets with varying workload levels.





\textbf{Performance across clusters of different scales:} In industrial settings, cluster sizes often fluctuate over time in response to evolving workload demands \cite{zhao2019minimal}. Assessing the scalability of our strategy across varying cluster sizes is therefore critical. The results presented in Fig.~\ref{Avg.Bandwidth_vs_Node_Number} demonstrate that our strategy sustains a consistent performance advantage across compared with other OCS-based cluster designs diverse cluster scales.

\subsection{Availability Analysis}
\begin{figure}[htbp]
    \centering 

    \includegraphics[width=0.5\textwidth]{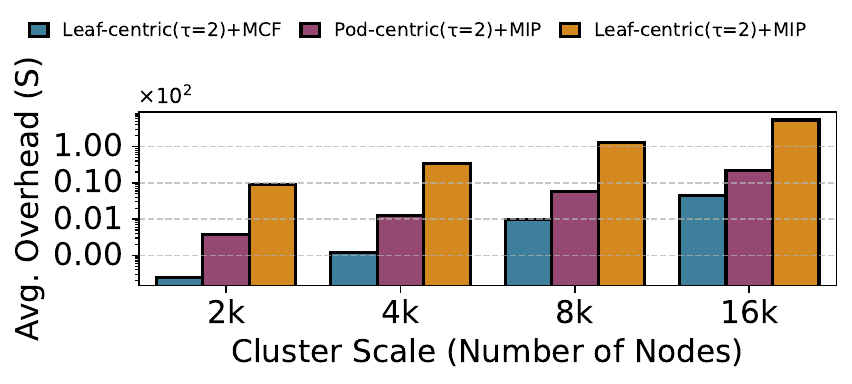}

    \caption{Comparative analysis of average time overhead in logical topology design}

    \label{fig:overhead_scale}
\end{figure}

Unlike traditional DCN workloads, ML workloads often require job-level reconfiguration \cite{wang2022topoopt}, resulting in JCT being directly tied to the overhead of logical topology computation. As shown in Fig.~\ref{fig:overhead_scale}, we evaluate and compare the computational overhead of logical topology derivation for our proposed LumosCore, alongside the MIP-based Pod-centric and leaf-centric approaches, across different cluster scales. When the cluster size reaches 16K, the MIP-based leaf-centric logical topology approach incurs an average computation time of 541.76 seconds, whereas our approach incurs merely 4.57 seconds on average. This 99.16\% reduction in computational overhead highlights the efficiency and efficacy of our algorithmic design. Our proposed logical topology design algorithm achieves a 6.26\% reduction in average JCT for an 8k-scale cluster, in comparison with the MIP-based leaf-centric logical topology design approach.

\section{Related Work}
\subsection{Handle the hash polarization problem}

\textbf{Hash Polarization} mitigation has emerged as a prominent research direction in GPU cluster network design, with numerous works addressing this challenge through diverse technical avenues. For instance, vClos \cite{han2023isolated,han2025vclos} mitigates communication bottlenecks by analyzing the traffic characteristics of communication primitives and redesigning the underlying algorithms accordingly. Alibaba’s HPN and other works \cite{qian2024alibaba,wang2024railonlylowcosthighperformancenetwork} minimizes cross-leaf traffic demand via a rail-optimized architectural design. ACCL \cite{dong2021accl} tackles traffic contention from the load-balancing perspective, leveraging multi-choice hashing techniques to select the least congested paths dynamically. Our work, by contrast, focuses on mitigating the \textbf{Routing Polarization} problem specifically induced by OCS deployments in GPU clusters. This unique focus makes our approach complementary to existing solutions, and it can be seamlessly integrated with these methods handling hash polarization to further enhance network performance.

\subsection{Designing of Current OCS-based GPU Cluster}

The first large-scale commercial deployment of MEMS-OCS in production clusters was documented in Google’s MinRewiring study, published in 2019 \cite{zhao2019minimal}. The primary objective of MinRewiring was to reduce the rewiring overhead associated with dynamic cluster scaling by exploiting the reconfigurability of MEMS-OCS. Building on this line of work, Google’s Jupiter Evolving architecture, presented in 2022 \cite{poutievski2022jupiter}, introduced a three-tier OCS-based topology. This design satisfies the traffic requirements of DCN workloads through traffic-aware logical topology construction and systematic topology engineering. With the rapid emergence of AI-oriented clusters, Google further employed similar architectures in its A3 and A4 \cite{GoogleA3Supercomputers,google-ai-hypercomputer-schedule-gke,google-kubernetes-ai-infra-integration} GPU clusters, leveraging the reconfigurable properties of OCS to accommodate the communication patterns and bandwidth demands of ML workloads. Simultaneously, companies such as NVIDIA \cite{Patronas:25} and Huawei \cite{huawei-optical-switch-intelligent-computing-2025} have also initiated research and development of GPU clusters optimized for ML workloads using analogous OCS-based architectural principles. 

Prior studies, including TPUv4 \cite{2023TPU}, TopoOpt \cite{wang2022topoopt}, MixNet \cite{Liao_2025}, and InfiniteHBD \cite{shou2025infinitehbdbuildingdatacenterscalehighbandwidth}, also investigate the design of OCS-based GPU cluster architectures, yet they primarily focus on systems with a limited number of network layers. Consequently, these works do not encounter the routing polarization issue that emerges in the three or more tier OCS-based GPU clusters, the primary focus of this study. However, when these approaches are integrated with RDMA (Remote Direct Memory Access)-enabled networks \cite{qian2024alibaba} for cluster scaling (e.g., as exemplified by Google’s presentation at OFC 2025 \cite{ofc2025_m4h6}), these issues become non-negligible and thus demand careful consideration.

While companies such as Google \cite{GoogleA3Supercomputers,huawei_2024_all_optical_switch,Patronas:25,google-ai-hypercomputer-schedule-gke,huawei-optical-switch-intelligent-computing-2025} have designed and even commercially deployed the three-tier OCS-based GPU clusters, they have \textbf{not mentioned} the problem of the routing polarization issue that this paper seeks to resolve. In this paper, we introduce a leaf-centric logical topology design methodology for OCS-based clusters. The proposed approach is specifically tailored to alleviate traffic contention on leaf-to-spine links, thereby improving the overall performance and scalability of OCS-based deployments.

Works like Interleaved Wiring \cite{zhao2021understanding,han2024lumoscore} have investigated the design of inter-Pod physical topologies; however, the design of intra-Pod physical topologies and strategies to mitigate routing polarization remain open for further investigation. Our work can be integrated with these studies to guide the deployment of OCS-based GPU clusters.

\subsection{The traffic pattern of ML workloads}\label{traffic_pattern}
According to publicly available information, Google has adopted OCS in three-tier OCS-based GPU clusters \cite{GoogleA3Supercomputers}. In contrast to DCN workloads, AI workloads are typically characterized as highly skewed yet structurally regular due to the design of their underlying communication algorithms~\cite{wang2022topoopt,Patronas:25}. This inherent regularity provides an opportunity to engineer OCS-based GPU clusters in a principled manner \cite{Liao_2025,Patronas:25}.

A representative ML training workload can be abstracted under the Tensor Parallelism–Pipeline Parallelism–Data Parallelism–Expert Parallelism (TP–PP–DP–EP) paradigm as Fig.\ref{fig:megatron} shows \cite{shoeybi2019megatron}. Among these, Tensor Parallelism (TP) is usually confined within a single server (or node) because of its extremely high communication volume, which makes intra-server communication preferable~\cite{DGX2023}. Expert Parallelism (EP) is commonly designed to be locality-aware and can be decomposed into multiple relatively independent communication domains~\cite{Liao_2025}. Although the all-to-all traffic induced by Mixture-of-Experts (MoE) layers in EP does not strictly follow a perfectly regular pattern~\cite{Liao_2025}, EP communication domains are typically small. Therefore, in works such as MixNet \cite{Liao_2025}, OCS reconfiguration can be exploited to accommodate the communication requirements of the EP communication domain.

\begin{figure}
    \centering
    \includegraphics[width=1\linewidth]{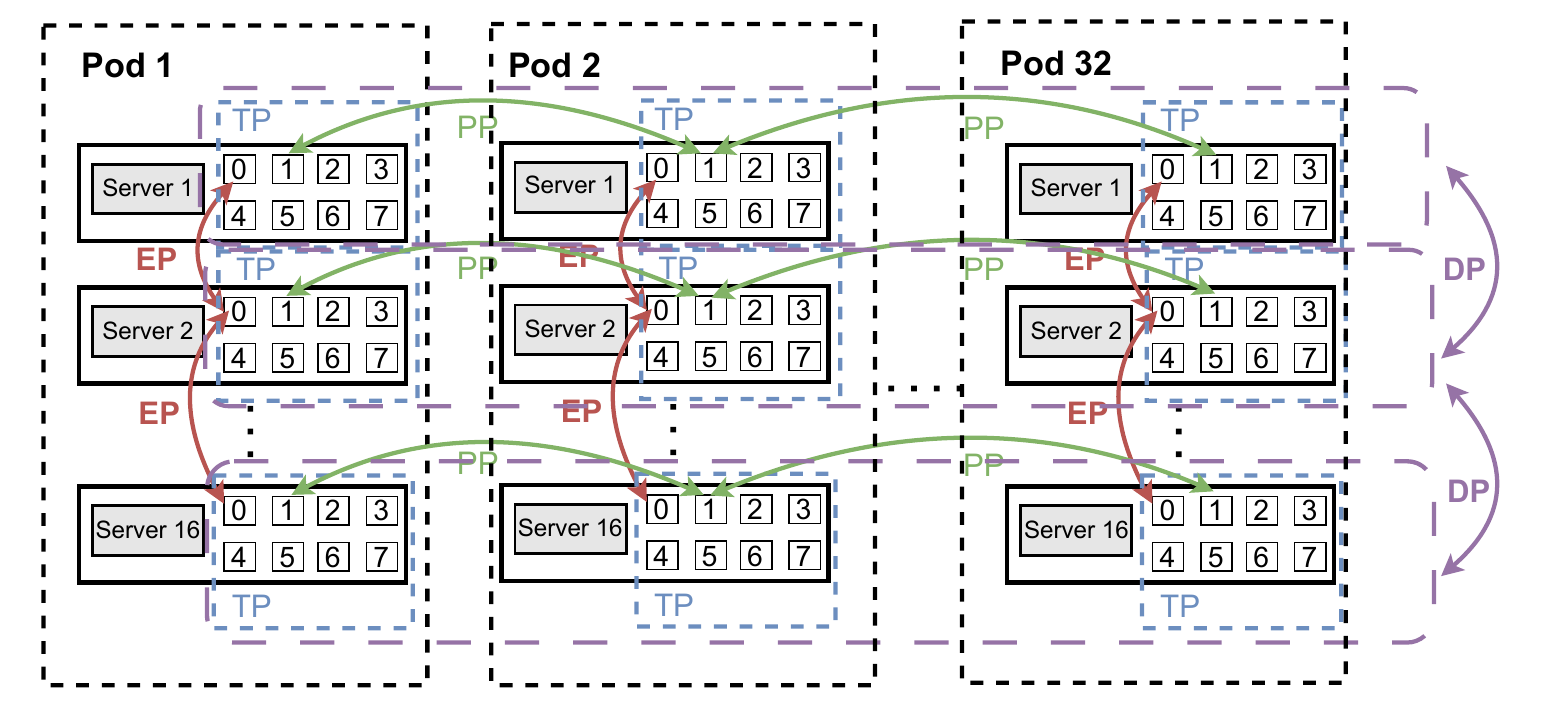}
    \caption{Communication Pattern for \textbf{Megatron} \cite{shoeybi2019megatron} with 4096 GPUs}
    \label{fig:megatron}
\end{figure}

Data Parallelism (DP) and Pipeline Parallelism (PP)~\cite{shoeybi2019megatron} communications typically exhibit stronger and more predictable regularity. This regularity enables the construction of suitable leaf-level network requirement for the corresponding communication domains prior to job execution~\cite{shou2025infinitehbdbuildingdatacenterscalehighbandwidth,wang2022topoopt}.Overall, informed by existing industrial deployments such as those of Google, we infer that designing OCS-based GPU clusters is a technically viable and practically deployable solution.

\section{Discussion}
\noindent\textbf{Generating \emph{Leaf-level Network Requirement}:}
While this paper primarily focuses on the details of logical topology design, we have not elaborated on the generating of the \emph{Leaf-level Network Requirement Matrix}. In fact, to minimize inter-leaf link demands as much as possible, it is essential to develop sophisticated communication algorithms and communication domain orchestration mechanisms. These critical components will be thoroughly detailed in our future work, where we will systematically address their design principles, implementation frameworks, and performance implications for large-scale GPU cluster architectures.

\section{Conclusion}

In this paper, we investigate logical topology design strategies for OCS-based GPU clusters and analyze how to design the intra-pod physical topology. Focusing on the previously overlooked yet inevitable issue of routing polarization in three-tier OCS-based GPU clusters, we develop a polynomial-time logical topology design strategy based on theoretical analysis. Furthermore, we propose an intra-pod physical topology design methodology that guarantees, for any given leaf-level traffic matrix, the computation in polynomial time of a logical topology that avoids these routing polarization issue.

\end{sloppypar}

\bibliographystyle{IEEEtran}
\bibliography{bibliography.bib}
\clearpage

\end{document}